%% file: main.tex
\documentclass{article}

\usepackage{arxiv}

\usepackage[utf8]{inputenc} 
\usepackage[T1]{fontenc}    
\usepackage{hyperref}       
\usepackage{url}            
\usepackage{booktabs}       
\usepackage{amsfonts}       
\usepackage{amssymb}
\usepackage{amsthm}
\usepackage{nicefrac}       
\usepackage{microtype}      
\usepackage{units}
\usepackage{lipsum}
\usepackage{subcaption}
\usepackage[misc]{ifsym}
\newcommand*\rot{\rotatebox{90}}
\usepackage{graphicx}
\newtheorem{theorem}{Theorem}[section]

\graphicspath{ {./images/} }

\title{New heuristics for burning graphs}
\author{
 Zahra Rezai Farokh \\
  Computer Science Department\\
  Shahid Beheshti University\\
  Tehran, Iran \\
  \texttt{rezaifarokhz@gmail.com} \\
  \And
 Maryam Tahmasbi \Letter \\
  Computer Science Department\\
  Shahid Beheshti University\\
  Tehran, Iran \\
  \texttt{m\_tahmasbi@sbu.ac.ir} \\
  \And
 Zahra Haj Rajab Ali Tehrani \\
  Computer Science Department\\
  Shahid Beheshti University\\
  Tehran, Iran \\
  \texttt{doorsatehrani@yahoo.com} \\
  \And
 Yousof Buali \\
  Computer Science Department\\
  Shahid Beheshti University\\
  Tehran, Iran \\
  \texttt{yousof.buali@gmail.com} \\
}

\begin{document}
\maketitle
\begin{abstract}
The concept of graph burning and burning number ($bn(G)$) of a graph
$G$ was introduced recently \cite{bonato2014burning}. Graph burning
models the spread of contagion(fire) in a graph in discrete time
steps. $bn(G)$ is the minimum time needed to burn a graph $G$. The
problem is NP-complete.

In this paper, we develop first heuristics to solve the problem in
general (connected) graphs. In order to test the performance of our
algorithms, we applied them on some graph classes with known burning
number such as $\theta-$graphs, we tested our algorithms on DIMACS
and BHOSLIB that are known benchmarks for NP-hard problems in graph
theory. We also improved the upper bound for burning number on
general graphs in terms of their distance to cluster. Then we
generated a data set of 2000 random graphs with known distance to
cluster and tested our heuristics on them.
\end{abstract}

\keywords{burning number, heuristic, distance to cluster, theta graphs,
DIMACS, BHOSLIB}

\include{introduction}
\include{algorithms}
\include{experimental_study}
\include{conclusion}

\bibliographystyle{unsrt}
\bibliography{references}
\end{document}

%% file: introduction.tex
\section{Introduction}
\label{sec:introduction} Burning number of a graph is a new concept
that measures the speed of spreading a contagion (fire) in a graph \cite{bonato2014burning}. Given an undirected unweighted network,
the fire spread in the network synchronous rounds as follows: in
round one, a fire starts at a vertex called an activator. In each
following round two events happen:
\begin{enumerate}
   \item The fire spreads to all neighbors of nodes that are on fire.
   \item Fire starts at a new activator that is an unburned vertex.
\end{enumerate}
The process continues until all the vertices of the graph are on
fire. At this time we say that the burning process is
complete \cite{kamali2020burning}. A burning schedule specifies a
burning sequence of vertices where the $i$th vertex in the sequence
is the activator in round $i$. The burning number $bn(G)$ is the
minimum length of a burning sequence.

Problem: burning number

Input: a simple graph $G$ of order $n$ and an integer $k \geq 2$.\\
Question: is $bn(G) \leq k$? In other words, does $G$ contain a
burning sequence $(x_1, x_2,\ldots, x_k)$?

As the first result, some of the properties of this problem
including characterizations and bounds was presented in \cite{bonato2014burning,roshanbin2016burning}. Bonato et al. \cite{bonato2014burning} proved that the burning number of any
connected graph with $n$ vertices is at most $2\sqrt{n}-1$ and
conjectured that it is always at most $\lceil\sqrt{n}\rceil$. It is
proved that this problem is NP-complete even when restricted to
trees with maximum degree three, spider graphs and
path-forests \cite{bessy2017burning}. They developed polynomial-time
algorithms for finding the burning number of spider graphs and
path-forests if the number of arms and components, respectively, are
fixed. They also generated a polynomial-time approximation algorithm
with approximation factor 3 for general graphs \cite{bessy2017burning}. Bonato and Lidbetter \cite{bonato2019bounds} developed a 3/2-approximation algorithm for path forests (disjoint union of paths). There is another approximation algorithm with an approximation ratio of 2 for trees \cite{bonato2019approximation}.

In a recent study Kamali et al. \cite{kamali2020burning} considered
connected $n$-vertex graphs with minimum degree $\delta$. They
developed an algorithm that burns any such graph in at most
$\sqrt{\frac{24n}{\delta+1}}$ rounds. In particular, for graphs with
$\delta\in\theta(n)$, they proved that all vertices are burned in a
constant number of rounds. More interestingly, even when $\delta$ is
a constant that is independent of $n$, their algorithm answers the
graph-burning conjecture in the affirmative by burning the graph in
at most $\sqrt{n}$ rounds. Some results for different classes of
graphs are presented in Table \ref{tab:results}.

\v{S}imon et al. \cite{vsimon2019heuristics} developed some heuristics
for graph burning based on some centrality measures. They tested
their heuristics on limited number of networks.

\begin{table}[h]
\centering
    \caption{Results from previous works}
    \label{tab:results}
\begin{tabular}{|l|l|l|}
\hline
Graph Classes & Results & Reference  \\\hline
 trees(maximum degree 3) & NP-completeness & \cite{bessy2017burning}  \\
 spider graphs,path-forests& & \\
\hline
spider graphs and path-forests & polynomial time algorithms  & \cite{bessy2017burning} \\
\hline
trees &  2-approximation algorithm & \cite{bonato2019approximation} \\
\hline
 graph products & exact value & \cite{mitsche2018burning} \\
\hline
  petersen graph & exact value & \cite{sim2018burning} \\
\hline
  theta graph & exact value &\cite{liu2019burning} \\
\hline
 dense and tree-like graph & exact value & \cite{kamali2019burning} \\
\hline
 grid graph & exact value &\cite{bonato2018burning} \\
\hline
 graph with constant $\delta$ & algorithm with almost $\sqrt{\frac{24n}{\delta+1}}$ rounds & \cite{kamali2020burning} \\
\hline
 graph with pathlength $pl$  & algorithm with almost $\sqrt{d-1}+pl$  & \cite{kamali2020burning} \\
 and diameter $d$ & rounds & \\
\hline
 graphs with $\delta\in\theta(n)$ & proved all vertices are burned & \cite{kamali2020burning} \\
 & in a constant number of rounds & \\
\hline
\end{tabular}
\end{table}

In this paper, we develop new heuristic algorithms for solving graph
burning problem. As mentioned before, most of the studies on this
problem concern limited classes of graphs. Since the problem is
modeling the spread of contagion in a network, it is essential to
develop algorithms for solving the problem. We developed 6
heuristics for burning a graph. These heuristics differ in selecting
the first activator and also the order of selecting the following
activators.

Except for approximation algorithms \cite{bonato2019bounds,kamali2019burning,vsimon2019heuristics}, and
algorithm 1 in \cite{bonato2016burn} there are no official
algorithms for this problem, so to test the performance of our
algorithm, we used some theoretical results: we generated a random
class of theta graphs and a random class of graphs with known
distance to cluster and report the result of applying our algorithms
on these classes. We compared our results with exact values and
bounds reported in former studies. We also applied our algorithms on
various graphs in known data sets: DIMACS and Benchmarks with Hidden
Optimum Solutions for Graph Problems (BHOSLIB). These data sets
contain graphs with various sizes and structures and are benchmarks
for testing several NP-hard graph algorithms including but not
limited to the maximum clique problem, maximum independent set,
minimum vertex cover and vertex coloring.

This paper is organized as follows: in section \ref{sec:algorithms} we present some basic definitions and describe our heuristics. In section \ref{sec:experimental-study} we present the result of our experimental study on different data sets. And in section \ref{sec:conclusion} we state conclusions and future works.

%% file: algorithms.tex
\section{Algorithms}
\label{sec:algorithms} In this section, we present 6 heuristics for
solving graph burning problem. The output in each algorithm is a
burning sequence for the input graph $G$. First,  we review some
basic definitions and then we present our heuristics.

\subsection{Basic Definitions}
We review some basic definitions from graph theory
\cite{diestel2000graduate}. For a graph $G = (V,E)$, let $V(G)$ and
$E(G)$ denote the vertex set and edge set of $G$ respectively. We
use $n$ and $m$ to denote the number of vertices and edges in a
graph respectively. For a vertex $v\in V(G)$, $N(v)$ denotes the set
of vertices adjacent to $v$ and $N[v] = N(v)\cup\{v\}$ is the closed
neighborhood of $v$.

The distance $d(u, v)$ between two vertices $u$ and $v$ in a graph
$G$ is the number of edges in a shortest path from $u$ to $v$. Given
an integer $k$, $N_k[v]$ is the number of vertices with distance at
most $k$ of $v$. This set is called the $k$th neighborhood of $v$.
For a vertex $v$ in a graph $G$, the eccentricity of $v$ is defined
as $max\{d(u, v) | u \in V (G)\}$. The radius of $G$ is minimum
eccentricity over the set of all vertices in $G$. The diameter of
$G$ is the maximum eccentricity over the set of all vertices in $G$.
In other words, it is the distance between the farthest pair of
vertices in $G$. For a subset $X \subset V(G)$, the graph
$G[X]$ denotes the subgraph of $G$ induced by vertices of $X$.

The theta graph $\theta(l_1,\ldots,l_m)$ is a graph consisting of
$m$ pairwise internally disjoint paths with common endpoints and
lengths $l_1\leq\cdots\leq l_m$ \cite{eichhorn2000edge}.

\subsection{Heuristics}
In the proposed algorithms we have two steps: the first step is to
select the first candidate for burning. It seems essential since
this vertex will burn vertices in distance $bn(G)$ of the graph. So,
we need to select a vertex with a large set of vertices in
$N_{bn(G)}[v]$.

In the second step, we select the rest of the activators one by one.

Given a burning sequence $S=(x_1,x_2,\ldots,x_{bn(G)})$ of a graph
$G$, for each vertex $v$, there is a vertex $x_i$ in $S$ such that
$v$ is burned by a fire that is started in $x_i$, i.e.
$d(v,x_i)<d(v,x_j)$ for all $j \neq i$. we call $x_i$ the activator
of $v$.

To reduce the length of a burning sequence, it is good to select
activators such that each vertex has a unique activator.

We develop different heuristics based on different strategies for
the first and second steps.
\begin{enumerate}
    \item Based  on arguments in former paragraphs, we can choose the first activator from the center of the graph. The farthest vertex to this vertex is in distance $rad(G)$ of it. So, it seems that this vertex has a big $k$th neighborhood. We used this in step one of heuristics \emph{Ctr-Half dist.} and \emph{Ctr-Far dist.}.
    \item In each time $k$, for each unburned vertex $v$, we can calculate that in how many time steps this vertex will burn if we do not add any other activator. We call this “time-to-burn” of $v$ and denote it by $t^k(v)$.  Let $t=max\{t^k(v):v\in V\}$. Hence, $t$ is the maximum number of remaining time to burning the whole graph. We can select the next activator in two ways:
    \begin{enumerate}
        \item The next activator is a vertex $v$ with $t^k(v)=t/2$. In this way the vertices with greater time-to-burn will burn in shorter time, using this new activator.
        \item The next activator is a vertex with max -1 time-to-burn.
    \end{enumerate}
    Figure \ref{fig:image2} shows the two strategies. Heuristics \emph{Ctr-Half dist.} and \emph{Rnd-Half dist.} use the first and heuristics \emph{Ctr-Far dist.} and \emph{Rnd-Far dist.} use the second way for this step.
    \item In heuristics \emph{Rnd-Half dist.} and \emph{Rnd-Far dist.} we select the first activator in random to see the effect of selecting the first activator in our heuristics.
\end{enumerate}

\begin{figure}[h]
\begin{subfigure}{0.5\textwidth}
\includegraphics[width=0.9\linewidth, height=0.8\linewidth]{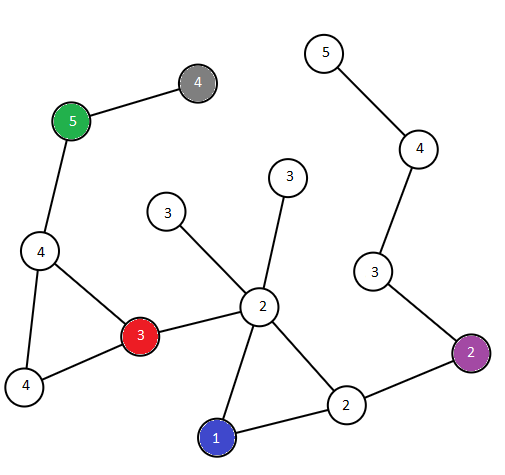}
\caption{Ctr-Half dist} \label{fig:subim1}
\end{subfigure}
\begin{subfigure}{0.5\textwidth}
\includegraphics[width=0.9\linewidth, height=0.8\linewidth]{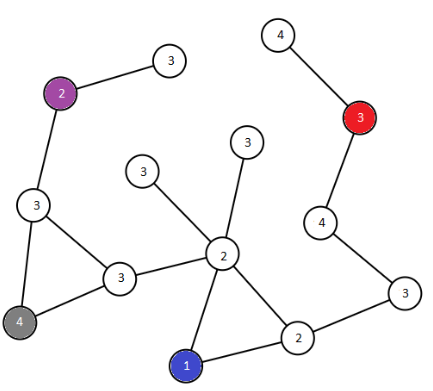}
\caption{Ctr-Far dist} \label{fig:subim2}
\end{subfigure}

\caption{Comparison of Ctr-Half dist and Ctr-Far dist. Activators
are colored vertices and the number inside each vertex shows the
time step it burns.} \label{fig:image2}
\end{figure}

Table \ref{tab:firstfour} summarizes the strategies in four
heuristics.

\begin{table}[h]
    \centering
    \caption{Summary of first four heuristics}
    \label{tab:firstfour}
    \begin{tabular}{|l|l|l|}
    \hline
        Heuristics & First Activator & Next Activator \\\hline
        Ctr-Half dist. & Center & \nicefrac{1}{2} time-to-burn \\\hline
        Ctr-Far dist. & Center & max time-to-burn \\\hline
        Rnd-Half dist. & Random & \nicefrac{1}{2} time-to-burn \\\hline
        Rnd-Far dist. & Random & max time-to-burn \\\hline
    \end{tabular}
\end{table}

We developed two other heuristics with a different idea: burning a
path! The main idea is finding the diameter of the graph and the
longest shortest path (the path with length $diam(G)$) and then
burning the vertices of this path with the same order as burning a
path in $\sqrt{diam(G)}$ steps. Since computing the diameter of a
graph is of large complexity, we use two approaches to find a good
approximation of that. In heuristic \emph{DFS-path} we select a
random vertex and perform a DFS algorithm to find a path.  In
heuristic \emph{D-BFS-path} we use the algorithm by Birmele et al.
\cite{birmele2017decomposing} to approximate the diameter. This
algorithm uses BFS twice, the first BFS starts from a random vertex
and the second one starts from one of the leaves of previous BFS.
This gives a 2-approximation of the diameter of the graph. There is
no guarantee that all vertices of the graph burn using only vertices
of these paths. So, after burning the vertices of the path, if there
is still an unburned vertex, we select them randomly as activators.

%% file: experimental_study.tex
\section{Experimental Study}
\label{sec:experimental-study}
In this section, we describe the software platform, hardware, data and results considered in this matter. We implemented algorithms that were introduced and explained earlier in section \ref{sec:algorithms} using Python 3. To model our graphs in a proper data structure and apply fundamental graph algorithms and measures, we used the well-known NetworkX package introduced by Hagberg et al. \cite{hagberg2008exploring} back in 2008.


These algorithms were evaluated on a hardware specification consisting of 4 units of Intel\textregistered Xeon\textregistered Processor E5-2680 v4 (56 cores) and 32GB of memory. We have not seen a critical need to parallelize the algorithms except for data. Instances of graphs were being evaluated in parallel using a simple data pool technique. We will explain our evaluation challenges later in this section.

\subsection{Datasets}\label{sec:dataset}
As mentioned before, there is no algorithm for burning general
graphs. So, in order to evaluate our heuristics we use two types of
datasets:
\begin{enumerate}
    \item Classic datasets that are commonly used in some NP-hard problems in graph theory such as clique number, independence set, dominating number, etc.
    \item Random graphs in some classes that the exact value or a good bound on their burning number is computed before.
\end{enumerate}
In this section, we describe these datasets in more detail.

\subsubsection{Classic datasets}
These datasets are prepared for public use in the Network Repository website\footnote{\url{http://networkrepository.com/}} \cite{nr}. We use graphs in DIMACS and BHOSLIB datasets. DIMACS dataset has 78 large graphs with a maximum of 4000 vertices and over 5 million edges. BHOSLIB dataset has 36 large graphs with a maximum of 4000 vertices and over 7 million edges.

\subsubsection{Random $\theta$-graphs}
A $\theta-$graph consists of 3 (internally) disjoint paths between two vertices. In other words, a $\theta-$graph is a cycle with a disjoint path joining two vertices of it. We use the following simple technique to generate a random theta graph:
\begin{enumerate}
    \item Generate a cycle $C_m$ with a random number of vertices, $m$.
    \item Choose two random vertices from the cycle.
    \item Generate a path $P_l$ of random size $l$ to join these vertices.
\end{enumerate}
The number of vertices for cycle and path can be chosen either randomly or purposely.

\subsubsection{Random graphs with fixed distance to cluster}

Another class of graphs with known upper bound on burning number is
the class of graphs with fixed distance to cluster. Distance to
cluster is the minimum number of vertices that have to be deleted
from $G$ to get a disjoint union of complete graphs. This is an
intermediate parameter between the vertex cover number and the
clique-width/rank-width \cite{kare2019parameterized}.

We generated 1000 random graphs in this class. Each of these graphs
consists of $k$ complete graphs with a random size that is connected
to a path of length $d$ with a random number of edges. In fact, this is
a special class of these graphs. We used the path structure of the $d$
vertices, since paths have known burning number of $\sqrt{d}$ and
this seems to be the maximum size in connected graphs. We generated
these graphs using the following steps:
\begin{enumerate}
    \item Generate a path $P_d$ of length $d$.
    \item Generate a set of $k$ complete graphs with sizes $K=\{K_{n_i}~|~i\in[1,k]\}$.
    \item For each graph in $K$ choose a random number $a_i$, $1\leq~a_i<n_i-1$ and choose $a_i$ random pairs of vertices $(u,v)$ such that $u\in~P_d$ and $v\in~K_{n_i}$.
    \item Connect $P_d$ and complete graphs in $K$ with pairs of vertices chosen above.
\end{enumerate}

Again our code enables $d$ and $n_i$s to be chosen either randomly or purposely.

\subsection{DIMACS, BHOSLIB}

We applied our 6 heuristics on all graphs in DIMACS and BHOSLIB.
From 78 graphs in DIMACS, all heuristics computed a burning sequence
of length 3 for 71 graphs. Theorem \ref{thm:roshanbin} shows that
this is optimal.

\begin{theorem}
\label{thm:roshanbin}
Let $G$ be a graph with $n$ vertices. Then $bn(G)=2$ if and only if
$n\geq 2$ and $G$ has maximum degree $n-1$ or $n-2$ \cite{roshanbin2016burning}.
\end{theorem}

Table \ref{tab:dimacs} shows the results for those DIMACS instances
with a burning sequence larger than 3. In these graphs, the average
vertex degree and max degree are very near and the number of edges
is small. This seems to be the result of the growing burning number.
Graphs c-fat200-* have 200 vertices and average degree grow from 15
to 84. As the average degree increases, the burning number decreases
and the results different heuristics converge. The same is true for
c-fat500-* graphs. On this benchmark, as the average degree grows,
the heuristics find shorter burning sequences.

\begin{table}[h]
    \centering
    \caption{DIMACS Results}
    \label{tab:dimacs}
    \begin{tabular}{|l|c|c|c|c|c|c|c|c|c|c|}
        \hline
        Name & \rot{Vertices} & \rot{Edges} & \rot{Max deg.} & \rot{Avg. deg.} & \rot{Ctr-Half} & \rot{Ctr-Far} & \rot{Rnd-Half} & \rot{Rnd-Far} & \rot{DFS-path} & \rot{D-BFS-path} \\\hline
        c-fat200-1 & 200 & 1534 & 17 & 15 & 11 & 8 & 9 & \textbf{7} & 8 & 8 \\\hline
        c-fat200-2 & 200 & 3235 & 34 & 32 & 6 & 6 & 6 & \textbf{5} & \textbf{5} & 6 \\\hline
        c-fat200-5 & 200 & 8473 & 86 & 84 & 4 & 4 & 4 & 4 & 4 & \textbf{3} \\\hline
        c-fat500-1 & 500 & 4459 & 20 & 17 & 12 & 11 & 12 & \textbf{10} & 15 & 17 \\\hline
        c-fat500-10 & 500 & 46627 & 188 & 186 & 4 & 4 & 4 & 4 & 4 & 4 \\\hline
        c-fat500-2 & 500 & 9139 & 38 & 36 & 9 & \textbf{8} & 9 & \textbf{8} & 11 & \textbf{8} \\\hline
        c-fat500-5 & 500 & 23191 & 95 & 92 & 6 & 6 & 6 & 6 & \textbf{5} & 6 \\\hline
    \end{tabular}
\end{table}

The results on BHOSLIB graphs are more interesting. All heuristics
compute 3 for all graphs. We note that the average degree is very
large compared with the number of vertices in these graphs. Since
maximum vertex degree in these graphs is less than $n-2$, by theorem
\ref{thm:roshanbin}, the burning sequence of length 3 is optimum for
all of them.


\subsection{$\theta-$Graphs}
This set of graphs is generated by our very own algorithms as
described earlier in this section.  We use a simple name convention
to understand the characteristics of each instance at a glance.
Graph names start with "theta" followed by some dash-separated
numbers, the number of vertices, sample and number of cycle and path
vertices, respectively. Our evaluation observed results for 2000
$\theta-$graphs ranging from 400 to 900 vertices. There are tight
bounds on the burning number of $\theta-$graphs that are proved by
Liu and et al. \cite{liu2019burning}. They showed that the burning
number of order $n=q^2+r$ with $1\leq r \leq 2q+r$ is either $q$ or
$q+1$. We compared our results with these bounds. In 1208 graphs
(\%60.4) the length of burning sequences in our heuristics meet
bounds and in \%81.7 of graphs, the difference is only one. The
average difference between our best results and upper bounds is 0.6
and standard deviation 1.2309.

Table \ref{tab:thetacomparison} shows the comparison of different
heuristics. DFS-path finds the shortest burning sequence in 1476
graphs that is more than 73.8\% of our graphs.

\begin{table}[h]
    \centering
    \caption{Comparison of different heuristics on theta graphs}
    \label{tab:thetacomparison}
    \begin{tabular}{|l|r|}
        \hline
        Heuristic & Success Rate \\\hline
        Ctr-Far & 16.9\% \\\hline
        Rnd-Far & 7.7\% \\\hline
        DFS-path & 73.8\% \\\hline
        D-BFS-path & 1.6\% \\\hline
    \end{tabular}
\end{table}

Table \ref{tab:theta} shows some randomly chosen samples of our
results. The expected burning number column is included in the
table.

\begin{table}[h]
    \centering
    \caption{Samples of $\theta-$graph Results}
    \label{tab:theta}
    \begin{tabular}{|l|c|c|c|c|c|c|c|}
        \hline
        Name & \rot{Expected\cite{liu2019burning}} & \rot{Ctr-Half} & \rot{Ctr-Far} & \rot{Rnd-Half} & \rot{Rnd-Far} & \rot{DFS-path} & \rot{D-BFS-path} \\\hline
        theta529-74-472-57 & \textbf{24} & 28 & 26 & 30 & 26 & \textbf{24} & 29 \\\hline
        theta784-99-493-291 & \textbf{29} & 33 & 32 & 37 & 32 & \textbf{29} & 51 \\\hline
        theta676-82-647-29 & \textbf{27} & 34 & 30 & 33 & 30 & \textbf{27} & 43 \\\hline
        theta676-115-163-513 & \textbf{27} & 35 & 28 & 36 & 31 & 31 & 34 \\\hline
        theta676-10-570-106 & \textbf{27} & 33 & 29 & 33 & 29 & 28 & 36 \\\hline
    \end{tabular}
\end{table}

\subsection{Graphs with fixed distance to cluster}
This set of graphs is also generated randomly and has a name
convention including characteristics of instances. Graph names start
with "cluster" which reminds us of the graph type followed by the
number of clusters ($K_n$s), minimum and maximum of random cluster
size, number of vertices on $P_d$, number of the whole graph
vertices and sample respectively. For this class of graphs, we
evaluated our algorithms for 1000 instances varying from 50 to 100
clusters of size 4 to 20 and path length from 500 to 1000.

Kare et al. \cite{kare2019parameterized} computed an upper bound for
graphs in terms of their distance to cluster, which is $3d+3$. We
improve this bound in the following theorem.

\begin{theorem}\label{thm:cluster}
Let $G$ be a graph and $A$ be a set of vertices such that
$G[V(G)\backslash A]$ is a cluster graph. Then $bn(G)\leq bn(G[A])+2$.
\end{theorem}

\begin{proof}
A burning sequence of $A$ burns all vertices except
possibly vertices of complete graphs that are adjacent to the last
vertex of the burning sequence. These complete graphs burn in at
most 2 rounds. So $bn(G)\leq  bn(G[A])+2$.
\end{proof}

An immediate conclusion from theorem \ref{thm:cluster} is that
$bn(G)\leq d+2$ for each graph with distance to cluster $d$.

\begin{table}[h]
    \centering
    \caption{Samples of random graphs and their distance to cluster}
    \label{tab:cluster}
\begin{tabular}{|l|c|c|c|c|c|c|c|c|}
    \hline
    Name & \rot{Distance to cluster} & \rot{upper bound} & \rot{Ctr-Half} & \rot{Ctr-Far} & \rot{Rnd-Half} & \rot{Rnd-Far} & \rot{DFS-path} & \rot{D-BFS-path} \\\hline
cluster52-4-20-592-1170-0315 & 592 & 27 & \textbf{11} & \textbf{11} & 12 & 13 & 13 & 14 \\\hline
cluster78-4-20-571-1496-0210 & 571 & 26             &11 & \textbf{10} & 11 & 11 & 13 & 12 \\\hline
cluster61-4-20-723-1498-0636 & 723 & 29 & 14 & \textbf{11} & 15 & 12 & 14 & 14 \\\hline
cluster10-4-20-759-886-0613 & 759 & 30 & 33 & \textbf{30} & 31 & 34 & 32 & 44 \\\hline
cluster13-4-20-861-1027-0447 & 861 & 32 & 32 & 29 & 33 & \textbf{28} & 32 & 33 \\\hline
cluster19-4-20-930-1190-0808 & 930 & 33 & 34 & \textbf{31} & 30 & 34 & 34 & 42 \\\hline
cluster73-4-20-512-1399-0496 & 512 & 25 & \textbf{10} & \textbf{10} & 12 & 12 & 13 & 13 \\\hline
\end{tabular}
\end{table}

As mentioned in section \ref{sec:dataset}, we generated random
graphs that each consists of a path of length $d$ and some complete
graphs that are connecting to this path with some edges. Since the
number of edges between each complete graph and the path is less
than the number of vertices of the complete graph, the distance to
cluster in this graph is $d$. Using theorem \ref{thm:cluster}, the
upper bound for burning number of these graphs is $\lceil\sqrt{d}\rceil+2$. Our data set consists of 2000 graphs. We applied
our heuristics on these graphs and compared the result with $\lceil\sqrt{d}\rceil+2$. The results show that in \%98 of graphs the
results meet bounds. Table \ref{tab:cluster} shows some random graphs and the burning number computed by each heuristic.

We also compared all heuristics. Heuristics Ctr-Half dist, Rnd-Half dist and Rnd-Far dist find better solutions among our 6 heuristics.

The winning heuristics that find the minimum solution in
1708 cases are the ones that select the first activator in different
ways and the following ones according to far dist. strategy.

%% file: conclusion.tex
\section{Conclusion and Future Work}
\label{sec:conclusion} In this paper, we developed the first
heuristics for graph burning problem. To study the performance of
our heuristics, we applied them on two types of datasets:
\begin{enumerate}
    \item Known benchmarks for NP-hard problems in graph theory. We selected DIMACS and BHOSLIB.
    Our heuristics computed the optimal solution in 71 graphs out of 78 graphs in DIMACS, and all the 36 graphs in BHOSLIB.
    \item Randomly generated graphs in classes with a known burning number. We generated 2000 $\theta-$graphs and applied our algorithms on
    them. Our heuristics succeeded to compute a burning sequence of
    length less than or equal to known upper bounds in 1208 graphs.
    \item There is an upper bound on the burning number in terms of distance to cluster. We improved this bound and generated a special
    class of 2000 random graphs where each graph is a path of length
    $d$ that is connected to a random number of disjoint complete graphs with
    some edges. In these graphs the distance to cluster is $d$ and
    we proved that their burning number is at most
    $\lceil\sqrt{d}\rceil+2$. In 1961 graphs, the burning number computed by our heuristics
    is less than or equal to this bound.
\end{enumerate}

Since there are very few studies on algorithmic approaches to solve
the burning number, of algorithmic approaches to solve the problem. We are
interested in other algorithmic approaches to solve the problem such
as local search algorithms. On the other hand, the problem has applications in social networks, that are usually disconnected graphs, we are going to develop heuristics for disconnected graphs.

Finally, there is a huge body of research on the spread of influence in
social networks. There are measures called centrality measures to
select seeds (activators) in a social network. It is interesting to
develop algorithms for burning graph using these measures.

We presented the first results on burning number of different data
sets of graphs. We hope that other algorithms will be developed for
this problem and compare their results with ours.